\documentclass[12pt,draftcls,onecolumn]{IEEEtran}

\usepackage{amsfonts}
\usepackage{dsfont}
\usepackage{mathrsfs}
\usepackage{amssymb}
\usepackage{bbm}
\usepackage{epsfig}
\usepackage{amsmath}
\usepackage{float}
\usepackage{color}
\usepackage[english]{babel}
\usepackage{cite}
\usepackage{subfigure}

\usepackage{graphicx}

\newtheorem{theorem}{Theorem}
\newtheorem{assumption}{Assumption}

\newtheorem{definition}{Definition}
\newtheorem{remark}{Remark}

\DeclareMathOperator{\sgn}{sign}
\usepackage{url}

\newenvironment{proof}[1][\bf Proof]%
{\smallskip\par\noindent\textit{#1: }}%
{\hspace*{\fill} \rule{6pt}{6pt}\smallskip}
\newenvironment{proof*}[1][Proof]%
{\smallskip\par\noindent\textbf{#1: }}
{\smallskip}

\allowdisplaybreaks

\def\R{\mathbb{R}}

\def\Z{\mathbb{Z}}

\newenvironment{system}[1]%
  {\setlength{\arraycolsep}{0.5mm}\left\{ \; \begin{array}{#1}}%
      {\end{array} \right.}

\title{Periodic Behaviors in Constrained  \\
Multi-agent Systems
\thanks{This work has been supported in part by the Knut and Alice Wallenberg Foundation, the Swedish Research
Council and EU HYCON2. A preliminary version of this paper will be presented at the 52nd IEEE Conference on Decision
and Control (CDC 2013), December 10-13, Florence, Italy. 
The contribution of this paper is two-fold:
1) we propose another constrained multi-agent model besides the one in \cite{yang-meng-dimarogonas-johansson-cdc13}
to produce the periodic behavior, and
2) we observe specific differences between these two models.}}

\author{Tao Yang, Ziyang Meng, Dimos V. Dimarogonas, and Karl H. Johansson
\thanks{The authors are with ACCESS Linnaeus Centre, School of Electrical Engineering, Royal Institute of Technology.
Stockholm 10044, Sweden. Email: {\tt\small $\{$taoyang, ziyangm, dimos, kallej$\}$@kth.se}}}

\begin{document}
\maketitle

\begin{abstract}
In this paper, we provide two discrete-time multi-agent models which generate periodic behaviors.
The first one is a multi-agent system of identical double integrators with input saturation constraints,
while the other one is a multi-agent system of identical neutrally stable system with input saturation constraints.
In each case, we show that if the feedback gain parameters of the local controller
satisfy a certain condition, the multi-agent system exhibits a periodic solution.
\end{abstract}


\section{Introduction} \label{sec-intro}

Generating sustainable oscillations in engineering systems is of fundamental importance, see e.g.,
\cite{chua-book,strogatz-book2,johansson-barabanov-astrom}. Many interconnected systems have a tendency to synchronize
their phase and frequency. Classical examples of oscillation multi-agent systems include fireflies
\cite{buck-firefly,strogatz-book}, Kuramoto oscillators \cite{kuramoto}, and Huygens' clock synchronization \cite{huygens}.

In the continuous-time setting, the existence of oscillatory behaviors in diffusively coupled
systems has been considered in \cite{pogromsky-glad-nijmeijer},
while the synchronization of Kuramoto oscillators have been studied in \cite{strogatz,dorfler-bullo}.
Periodic behaviors in diffusively coupled discrete-time systems seem to have been neglected in the literature.
In this paper, we propose two discrete-time constrained multi-agent models which generate periodic behaviors.
The agents in the network are identical, diffusively coupled and have input saturation constraints.
The identical agent model is a double integrator in the first case, and is a neutrally stable system in the other case.
In each case, we show that if the feedback gain parameters of the local controller satisfy a certain condition, then
the multi-agent system exhibits a periodic solution.
The contribution of this paper is two-fold:
1) we propose another constrained multi-agent model besides the one in \cite{yang-meng-dimarogonas-johansson-cdc13}
to produce the periodic behavior, and
2) we observe specific differences between these two models.
More specifically, the feedback gain parameters for achieving the
periodic behavior do not depend on the network topology in the double integrator case,
however they depend on the network topology in the neutrally stable case.
The period depends on the network topology in the double integrator, however it is independent of the network topology in the neutrally stable case.
This paper can also be viewed as an extension of the results in Theorem 21.9 and Corollary 21.10 of \cite{rugh} to multi-agent systems.


\section{Preliminaries and Notations}\label{sec-preliminary}
In this paper, we assume that the communication topology
among the agents is described by a fixed undirected weighted
graph $\mathcal{G}=(\mathcal{V},\mathcal{E},\mathcal{A})$,
with the set of agents $\mathcal{V}=\{1,\ldots, N\}$, the set of undirected edges
$\mathcal{E} \subseteq \mathcal{V} \times \mathcal{V}$,
and a weighted adjacency matrix $\mathcal{A}=[a_{ij}] \in {\R}^{N \times N}$,
where $a_{ij}>0$ if and only if $(j,i) \in \mathcal{E}$ and $a_{ij}=0$ otherwise.
We also assume that $a_{ij}=a_{ji}$ for all $i,j \in \mathcal{V}$ and that
there are no self-loops, i.e., $a_{ii}=0$ for $i \in \mathcal{V}$.
The set of neighboring agents of agent $i$ is defined as
$\mathcal{N}_i=\{j \in \mathcal{V}|a_{ij}>0\}$.
A path from node $i_{1}$ to $i_{k}$ is a sequence of nodes $\left\{i_1,\ldots,i_k\right\}$ such that $(i_j,i_{j+1}) \in \mathcal{E}$
for $j=1,\ldots,k-1$.
The unweighted distance between two nodes $i$ and $j$ denoted by $d(i,j)$ is the number of edges of a path between $i$ and $j$ minimized over all
such paths.
An undirected graph is said to be connected if there exists a path between any pair of distinct nodes.
A node is called a root if there exists a path to every other node. For a connected graph, every node is a root.

For an undirected weighted graph $\mathcal{G}$,
a matrix $L=\{\ell\}_{ij} \in {\R}^{N\times N}$ with $\ell_{ii}=\sum_{j=1}^{N}a_{ij}$ and $\ell_{ij}=-a_{ij}$ for $j \neq i$,
is called Laplacian matrix associated with the graph $\mathcal{G}$. It is well known that the Laplacian matrix has
the property that all the row sums are zero.
If the undirected weighted graph $\mathcal{G}$ is connected, then $L$ has a simple eigenvalue at zero with corresponding
right eigenvector ${\bf 1}$ and all other eigenvalues are strictly positive.
For such a case, all the eigenvalues of the Laplacian matrix
can be ordered as $0=\lambda_1<\lambda_2 \leq \ldots \leq \lambda_N$. 
The set of positive integers is denoted by $\Z^{+}$ while the set of non-negative integer is denoted by $\Z$.
$I$ denotes the identity matrix whose dimension can be deducible from the context.

\section{Two Multi-agent Models}\label{sec-prob}
In this paper, we propose two discrete-time multi-agent models which generate periodic behaviors.

The first model is a multi-agent system of $N$ identical
discrete-time double integrators described by
\begin{equation}\label{agent-model-di}
\begin{bmatrix}
x_i(k+1) \\
v_i(k+1)
\end{bmatrix}=\begin{bmatrix}
1 & 1 \\
0 & 1
\end{bmatrix}\begin{bmatrix}
x_i(k) \\
v_i(k)
\end{bmatrix}+\begin{bmatrix}
0 \\
1
\end{bmatrix}\sigma(u_i(k)), \quad i \in \mathcal{V},
\end{equation}
where $\sigma(u)$ is the standard saturation function:
$\sigma(u)=\sgn(u)\min\{1,|u|\}$.

The second model is a multi-agent system of $N$ identical
discrete-time neutrally stable systems described by
\begin{equation}\label{agent-model-ns}
\begin{bmatrix}
x_i(k+1) \\
v_i(k+1)
\end{bmatrix}=A \begin{bmatrix}
x_i(k) \\
v_i(k)
\end{bmatrix}+B \sigma(u_i(k)):=\begin{bmatrix}
0 & 1\\
-1 & 2a
\end{bmatrix} \begin{bmatrix}
x_i(k) \\
v_i(k)
\end{bmatrix}+\begin{bmatrix}
0 \\
1
\end{bmatrix} \sigma(u_i(k)), \quad i \in \mathcal{V},
\end{equation}
where $-1<a<1$ and $a \neq 0$.

\begin{remark}
The model \eqref{agent-model-ns} is not as restricted as it seems to be.
In fact, any planar neutrally stable system with single input channel,
such that the pair $(A,B)$ is controllable, and the matrix $A$ has all the eigenvalues on the unit circle
except $\pm 1$ and $\pm j$, can be transferred into \eqref{agent-model-ns}.
\end{remark}

We make following assumption about the network topology.
\begin{assumption}\label{ass-network}
The undirected graph $\mathcal{G}$ is connected.
\end{assumption}

Consider state feedback control laws
based on the agent state relative to that of neighboring agents
with feedback gain parameters $\alpha$ and $\beta$ of the form
\begin{equation}\label{controller-di}
u_i(k)=\alpha \sum _{j \in \mathcal{N}_i} a_{ij}(x_j(k)-x_i(k))+\beta \sum _{j \in \mathcal{N}_i} a_{ij}(v_j(k)-v_i(k)). 
\end{equation}

In this paper, we will examine the behavior of the multi-agent systems \eqref{agent-model-di}
and \eqref{agent-model-ns} under the distributed controller \eqref{controller-di} respectively,
and show that both multi-agent systems exhibit a periodic solution, defined as follows.

\begin{definition}\label{def-periodic}
The multi-agent system \eqref{agent-model-di} or \eqref{agent-model-ns} under the distributed controller
\eqref{controller-di} exhibits a periodic solution with a period $T>0$,
if for some initial states $x_i(0)$ and $v_i(0)$ for $i \in \mathcal{V}$, we have
$x_i(k+T)=x_i(k)$ and $v_i(k+T)=v_i(k)$ for all $k \in \Z $ and for all $i \in \mathcal{V}$.
\end{definition}

\section{Main Results}\label{sec-main}

For presenting our main results, we need to define the following sets
based on whether distance between agent $i \in \mathcal{V}$ and the root agent $1$
\footnote{Since the graph is connected,
without loss of generality, we assume that the agent $1$ is the root agent.}
is even or odd,
\begin{equation}\label{set-even-odd}
S_{e}=\{i | d(i,1)=2s \}, \, \text{ and } \, S_{o}=\{i | d(i,1)=2s+1\}, \quad s\in \Z.
\end{equation}
Let us also define
\begin{equation}\label{min-edge-weight-ns}
\bar{a}=\min_{\substack{(i,j) \in \mathcal{E} \\i \in S_e, \, j \in S_o}} a_{ij}.
\end{equation}
We are now ready to present our main results.

\subsection{Double Integrator Case}

For the multi-agent system \eqref{agent-model-di} under the distributed controller \eqref{controller-di}, we
have the following result.

\begin{theorem}\label{thm1}
Assume that Assumption \ref{ass-network} is satisfied.
Consider the multi-agent system \eqref{agent-model-di} under the distributed controller \eqref{controller-di}.
If the feedback gain parameters $\alpha$ and $\beta$ satisfy
\begin{equation}\label{period}
0< \alpha <\beta < \tfrac{3}{2}\alpha,
\end{equation}
then the multi-agent system exhibits a periodic solution with some period $T>0$.
\end{theorem}

\begin{proof}
Since the graph is connected, every agent is a root agent.
Without loss of generality, we assume that the agent $1$ is the root agent.

We shall prove the theorem by explicitly constructing
periodic solutions with an even period $T=2m$, for some $m \in \mathcal{Z}^{+}$, which
will be specified later in the proof.

Let the sets $S_e$ and $S_o$ be defined by \eqref{set-even-odd}.
The periodic solution that we will construct is such that all agents are always in saturation.
Moreover, the saturated input sequences are composed of $1$ for the first $m$
steps, followed by $-1$ for the next $m$ steps for agent $i \in S_e$,
and the saturated input sequences are composed of $-1$ for the first $m$
steps, followed by $1$ for the next $m$ steps for agent $i \in S_o$, that is,
\begin{equation}\label{group1}
\begin{system}{ccl}
& u_i(k) \geq 1,  & k=0,\ldots,m-1, \\
& u_i(k)  \leq -1,  & k=m,\ldots,2m-1,
\end{system}\quad \quad i \in S_e,
\end{equation}
\begin{equation}\label{group2}
\begin{system}{ccl}
& u_i(k) \leq -1, &  k=0,\ldots,m-1, \\
& u_i(k) \geq 1, & k=m,\ldots,2m-1,
\end{system}\quad \quad i \in S_o.
\end{equation}
In what follows, we will show that the above $2Nm$ inequalities are
satisfied for certain positive integer $m$ and initial states $x_i(0)$ and $v_i(0)$ for $i \in \mathcal{V}$,
and that moreover the multi-agent systems exhibits a periodic behavior with period $T=2m$
for these initial states in three steps.

{\bf Step 1:} Due to the required characteristic of the saturated input sequences, i.e., \eqref{group1} and
\eqref{group2}, we see that in order to have the periodic solution defined in Definition \ref{def-periodic},
it is sufficient to have $x_i(T)=x_i(0)$ and $v_i(T)=v_i(0)$ for all $i \in \mathcal{V}$.

It follows from \eqref{agent-model-di}, \eqref{group1} and \eqref{group2} that
for $k=1,\ldots,m$,
\begin{equation*}\label{ineq-di-1}
\begin{system}{ccl}
x_i(k) &=& x_i(0)+k v_i(0)+\tfrac{k(k-1)}{2}, \\
v_i(k) &=& v_i(0)+k,
\end{system} \quad i \in S_e,
\end{equation*}
and
\begin{equation*}\label{ineq-di-2}
\begin{system}{ccl}
x_i(k) &=& x_i(0)+k v_i(0)-\tfrac{k(k-1)}{2}, \\
v_i(k) &=& v_i(0)-k,
\end{system}\quad i \in S_o,
\end{equation*}
while for $k=m+1,\ldots,2m$,
\begin{equation*}\label{ineq-di-3}
\begin{system}{ccl}
x_i(k) &=& x_i(m)+(k-m) v_i(m)-\tfrac{(k-m)(k-m-1)}{2}, \\
v_i(k) &=& v_i(m)-(k-m),
\end{system}\quad i \in S_e,
\end{equation*}
and
\begin{equation*}\label{ineq-di-4}
\begin{system}{ccl}
x_i(k) &=& x_i(m)+(k-m) v_i(m)+\tfrac{(k-m)(k-m-1)}{2}, \\
v_i(k) &=& v_i(m)+(k-m),
\end{system}\quad i \in S_o.
\end{equation*}
Note that $v_i(T)=v_i(0)$ for all $i \in \mathcal{V}$.
It is also easy to obtain that
\[
\begin{system}{ccl}
x_i(2m)&=&x_i(0)+2m v_i(0)+m^2, \quad i \in S_e, \\
x_i(2m)&=&x_i(0)+2m v_i(0)-m^2, \quad i \in S_o.
\end{system}
\]
Thus, in order to have $x_i(T)=x_i(0)$ for all $i \in \mathcal{V}$, we must have that
\begin{equation}\label{initial-cond-v}
\begin{system}{ccl}
v_i(0)&=&-\frac{m}{2},  \quad i \in S_e. \\
v_i(0)&=&\frac{m}{2},  \quad  i \in S_o.
\end{system}
\end{equation}

{\bf Step 2:} In this step, we show that the $2m$ inequalities, either \eqref{group1} or \eqref{group2}
can be reduced to only two inequalities by appropriately choosing initial states $x_i(0)$ for some $i \in \mathcal{V}$.

{\bf Step 2.1:} Consider the input for an agent $j \in S_o$,
\begin{align}
u_j(k) &=\alpha \sum _{i \in \mathcal{N}_j} a_{ij}(x_i(k)-x_j(k))+\beta \sum _{i \in \mathcal{N}_j} a_{ij}(v_i(k)-v_j(k)) \nonumber \\
&=\alpha \sum _{i \in \mathcal{N}_j \cap S_e} a_{ij}(x_i(k)-x_j(k))+\beta \sum _{i \in \mathcal{N}_j \cap S_e} a_{ij}(v_i(k)-v_j(k)) \nonumber \\
&\hspace*{0.4cm}+\alpha \sum _{i \in \mathcal{N}_j \cap S_o} a_{ij}(x_i(k)-x_j(k))+\beta \sum _{i \in \mathcal{N}_j \cap S_o} a_{ij}(v_i(k)-v_j(k)). \label{keyineq2}
\end{align}
Choose $x_i(0)=x_j(0)$ for $i \in S_o$ if $(i,j) \in \mathcal{E}$.
By applying this and the fact that $v_i(0)=v_j(0)$ for all $i,j \in S_o$ implied by \eqref{initial-cond-v} to
\eqref{keyineq2}, we obtain
\begin{equation}\label{keyineq1}
u_j(k) =\alpha \sum _{i \in \mathcal{N}_j \cap S_e} a_{ij}(x_i(k)-x_j(k))+\beta \sum _{i \in \mathcal{N}_j \cap S_e} a_{ij}(v_i(k)-v_j(k)).
\end{equation}
Similarly, for an agent $i \in S_e$, choosing $x_j(0)=x_i(0)$ for $j \in S_e$ if $(i,j) \in \mathcal{E}$,
yields
\begin{equation*}
u_i(k) =\alpha \sum _{j \in \mathcal{N}_i \cap S_o} a_{ij}(x_j(k)-x_i(k))+\beta \sum _{j \in \mathcal{N}_i \cap S_o} a_{ij}(v_j(k)-v_i(k)).
\end{equation*}

{\bf Step 2.2:} Let us now focus on any edge  $(i,j) \in \mathcal{E}$, where $i \in S_e$ and $j \in S_o$.
We first note that $0<\alpha<\beta$ from \eqref{period}. This implies:
\[
\beta >\alpha - \tfrac{1}{2} k \alpha
\]
for $k=0,\ldots,m-1$. This yields
\[
-\tfrac{\alpha  m}{2}+\beta > \tfrac{1}{2} \alpha(-m-k+2).
\]
Since $m-k-1\geq 0$, multiplying the
above inequality on both sides with $m-k-1$ yields:
\[
- \tfrac{\alpha m}{2} (m-k-1)+\beta(m-k-1)
\geq \alpha \left[ \tfrac{k(k-1)}{2} - \tfrac{(m-1)(m-2)}{2} \right].
\]
This is equivalent to:
\begin{align}
& a_{ij} \Bigl\{ \alpha \left[x_i(m-1)-x_j(m-1)\right]+\beta \left[v_i(m-1)-v_j(m-1)\right] \Bigr\}  \nonumber \\
&\geq a_{ij} \Bigl\{ \alpha [x_i(k)-x_j(k)]+\beta [v_i(k)-v_j(k)]\Bigr\} \label{keyineq4}
\end{align}
for $k=0,\ldots,m-1$, since $v_i(0)=-\frac{m}{2}$ for $i \in S_e$, $v_j(0)=\frac{m}{2}$ for $j \in S_o$,
and $a_{ij} \geq 0$.

{\bf Step 2.3:} Since the inequality \eqref{keyineq4} holds for each $i \in \mathcal{N}_j \cap S_e$, where $j \in S_o$, then adding them up
and noting \eqref{keyineq1} yields,
\[
u_j(m-1) \geq u_j(k), \quad k=0, \ldots, m-1, \quad j \in S_o.
\]
Hence, for $j \in S_o$
\begin{equation}\label{keySo1}
u_j(m-1)\leq -1
\end{equation}
implies that $u_j(k) \leq -1$ for all $k=0,\ldots,m-1$.

A similar argument shows that if
\begin{align*}
& a_{ij} \Bigl\{ \alpha \left[x_i(2m-1)-x_j(2m-1)\right]+\beta \left[v_i(2m-1)-v_j(2m-1)\right] \Bigr\}  \nonumber \\
&\leq a_{ij} \Bigl\{ \alpha [x_i(k)-x_j(k)]+\beta [v_i(k)-v_j(k)]\Bigr\} 
\end{align*}
for each $i \in \mathcal{N}_j \cap S_e$, where $j \in S_o$, then adding them up
and noting \eqref{keyineq1} yields,
\[
u_j(2m-1) \leq u_j(k), \quad k=m, \ldots, 2m-1.
\]
Hence, for $j \in S_o$
\begin{equation}\label{keySo2}
u_j(2m-1)\geq 1
\end{equation}
implies that $u_j(k) \geq 1$ for all $k=m,\ldots,2m-1$.

Similarly, we can show that for $i \in S_e$,
\begin{equation} \label{keySe1}
u_i(m-1)\geq 1
\end{equation}
implies that $u_i(k) \geq 1$ for $k=0,\ldots,m-1$,
and that
\begin{equation}\label{KeySe2}
u_i(2m-1)\leq -1
\end{equation}
implies that $u_i(k)\leq -1$ for $k=m,\ldots,2m-1$.

To summarize, if there is an edge connecting agents within $S_e$ or $S_o$, we set their initial states the same, i.e.,
\begin{equation}\label{cyclic-cond}
x_i(0)=x_j(0)\, \text{ for } (i,j) \in \mathcal{E}, \text{ if } i,j \in S_e, \text{ or } i,j \in S_o,
\end{equation}
then the $2m$ inequalities for agent $i \in \mathcal{V}$, i.e., either \eqref{group1} or \eqref{group2},
are reduced to only two inequalities, i.e., inequalities \eqref{keySo1} and \eqref{keySo2} for $j \in S_o$, or
inequalities \eqref{keySe1} and \eqref{KeySe2} for $i \in S_e$.

{\bf Step 3:} It is clear that
if for each edge $(i,j) \in \mathcal{E}$, where $i \in S_e$ and $j \in S_o$
the following two conditions
\begin{align}
&\hspace{0.6cm} a_{ij} \Bigl\{ \alpha\left[x_i(m-1)-x_j(m-1)\right]+\beta\left[v_i(m-1)-v_j(m-1)\right]\Bigr\} \nonumber \\
&=a_{ij}\Bigl\{ \alpha [x_i(0)+(m-1)v_i(0)+\tfrac{(m-1)(m-2)}{2}-x_j(0)-(m-1)v_j(0)+\tfrac{(m-1)(m-2)}{2}]\nonumber \\
& \hspace*{0.6cm}+\beta \left[v_i(0)+m-1-v_j(0)+m-1\right] \Bigr\}\nonumber\\
&=a_{ij} \Bigl\{ \alpha \left[x_i(0)-x_j(0)-2m+2\right]+\beta(m-2)\Bigr\} \leq -1, \label{ineq-di-key1}
\end{align}
and
\begin{align}
&\hspace{0.6cm} a_{ij}\Bigl\{ \alpha\left[x_i(2m-1)-x_j(2m-1)\right]+\beta\left[v_i(2m-1)-v_j(2m-1)\right]\Bigr\} \nonumber\\
&=a_{ij}\Bigl\{\alpha \left[x_i(0)-v_i(0)-1-x_j(0)+v_j(0)-1\right]+\beta \left[v_i(0)+1-v_j(0)+1\right]\Bigr\} \nonumber\\
&=a_{ij}\Bigl\{\alpha \left[x_i(0)-x_j(0)+m-2 \right]-\beta(m-2)\Bigr\} \geq 1, \label{ineq-di-key2}
\end{align}
are satisfied for some initial states $x_i(0)$ and $x_j(0)$ where $i \in S_e$, $j \in S_o$, and $(i,j) \in \mathcal{E}$.
then \eqref{keySo1} and \eqref{keySo2} for $j \in S_o$, and \eqref{keySe1} and \eqref{KeySe2} for $i \in S_e$.

Two inequalities \eqref{ineq-di-key1} and \eqref{ineq-di-key2} are equivalent to
\begin{equation}\label{xxxx}
\tfrac{1}{a_{ij}}+(\beta-\alpha)(m-2) \leq \alpha (x_i(0)-x_j(0)) \leq 2\alpha(m-1)-\beta(m-2)-\tfrac{1}{a_{ij}},
\end{equation}
for each $(i,j) \in \mathcal{E}$, where $i \in S_e$ and $j \in S_o$.

We see that suitable $x_i(0)$ and $x_j(0)$ where $i \in S_e$, $j \in S_o$, and $(i,j) \in \mathcal{E}$
exist if and only if
\begin{equation}\label{keym}
\tfrac{1}{a_{ij}}+(\beta-\alpha)(m-2) \leq 2\alpha(m-1)-\beta(m-2)-\tfrac{1}{a_{ij}},%
\end{equation}
for each $(i,j) \in \mathcal{E}$, where $i \in S_e$ and $j \in S_o$.

For $m>2$, \eqref{keym} is equivalent to
\begin{equation*}
\beta \leq \tfrac{3m-4}{2m-4}\alpha -\tfrac{1}{a_{ij}(m-2)}.
\end{equation*}
If we take the value of $m$ to be very large, we obtain that
\[
\beta \leq \lim _{m \rightarrow +\infty}
\left[\tfrac{3m-4}{2m-4}\alpha-\tfrac{1}{a_{ij}(m-2)}\right]=\tfrac{3}{2}\alpha.
\]
Therefore for any $\alpha$ and $\beta$ which satisfy the condition \eqref{period},
if
\begin{equation}\label{suffm}
m \geq \tfrac{4(\alpha-\beta)+\tfrac{2}{\bar{a}}}{3\alpha-2\beta},
\end{equation}
where $\bar{a}$ is defined by \eqref{min-edge-weight-ns},
then all the inequalities \eqref{keym} are satisfied.

From the above analysis, we see that the multi-agent system \eqref{agent-model-di} under
the distributed controller \eqref{controller-di} with the feedback gain parameters
$\alpha$ and $\beta$ satisfying the condition \eqref{period} exhibits a periodic behavior with period
$T=2m$, where $m$ satisfies \eqref{suffm}, if initial states $x_\ell(0), \ell \in \mathcal{V}$
satisfy the conditions \eqref{initial-cond-v}, \eqref{cyclic-cond}, and \eqref{xxxx}.
\end{proof}

%

\begin{remark}\label{remark-double}
We note that the periodic behavior presented in Theorem \ref{thm1} exists for a particular set of
initial states as given by \eqref{initial-cond-v}, \eqref{cyclic-cond}, and \eqref{xxxx}.
The condition \eqref{period} on feedback gain parameters $\alpha$ and $\beta$ for achieving periodic behaviors
does not depend on the network topology. However, the period $T=2m$, where $m$ is given by \eqref{suffm}, depends on
the network topology and the feedback gain parameters.
\end{remark}

\subsection{Neutrally Stable System Case}
For the multi-agent system \eqref{agent-model-ns} under the distributed controller \eqref{controller-di}, we
have the following result.

\begin{theorem}\label{thm2}
Assume that Assumption \ref{ass-network} is satisfied.
Consider the multi-agent system \eqref{agent-model-ns} and the distributed controller \eqref{controller-di}.
If the feedback gain parameters  $\alpha$ and $\beta$ satisfy
\begin{equation}\label{condition-ns}
|\alpha| \leq \sgn (a) (\beta-\frac{a}{\bar{a}}),
\end{equation}
with $\bar{a}$ defined by \eqref{min-edge-weight-ns},
then the multi-agent system exhibits a periodic solution with period $T=4$.
\end{theorem}

\begin{proof}
Let the sets $S_e$ and $S_o$ be defined by \eqref{set-even-odd}.
The periodic solution with period $T=4$ is such that all agents are always in saturation.
Moreover,
\begin{equation}\label{ns-group1}
\begin{system}{ccl}
& u_i(k) \geq 1,  & k=0,1, \\
& u_i(k)  \leq -1,  & k=2,3,
\end{system} \quad  i \in S_e,
\end{equation}
\begin{equation}\label{ns-group2}
\begin{system}{ccl}
& u_i(k) \leq -1, &  k=0,1, \\
& u_i(k) \geq 1, & k=2,3,
\end{system}\quad i \in S_o.
\end{equation}
In what follows, we will show that the above $4N$ inequalities are
satisfied for certain initial states $x_i(0)$ and $v_i(0)$, $i \in \mathcal{V}$, and that
moreover the multi-agent system exhibits a periodic behavior with period $T=4$
for these initial states in three steps.

{\bf Step 1:} Due to the required characteristic of the saturated input sequences, i.e., \eqref{ns-group1} and \eqref{ns-group2}
and Definition \ref{def-periodic}, we see that in order to have the periodic solution defined in Definition \ref{def-periodic},
it is sufficient to have $x_i(T)=x_i(0)$ and $v_i(T)=v_i(0)$ for all $i \in \mathcal{V}$.

It follows from \eqref{agent-model-ns} and \eqref{ns-group1} that for $i \in S_e$,
\begin{align*}
\begin{bmatrix}
x_i(1) \\
v_i(1)
\end{bmatrix}&=A\begin{bmatrix}
x_i(0) \\
v_i(0)
\end{bmatrix}+B, \\
\begin{bmatrix}
x_i(2) \\
v_i(2)
\end{bmatrix}&=A^2\begin{bmatrix}
x_i(0) \\
v_i(0)
\end{bmatrix}+AB+B, \\
\begin{bmatrix}
x_i(3) \\
v_i(3)
\end{bmatrix}&=A^3\begin{bmatrix}
x_i(0) \\
v_i(0)
\end{bmatrix}+A^2B+AB-B, \\
\begin{bmatrix}
x_i(4) \\
v_i(4)
\end{bmatrix}&=A^4 \begin{bmatrix}
x_i(0) \\
v_i(0)
\end{bmatrix}+A^3 B+A^2B-AB-B.
\end{align*}
Therefore, in order to have
\[
\begin{bmatrix}
x_i(0)\\
v_i(0)
\end{bmatrix}=\begin{bmatrix}
x_i(4)\\
v_i(4)
\end{bmatrix}, \quad i \in S_e,
\]
we need
\[
\begin{bmatrix}
x_i(0) \\
v_i(0)
\end{bmatrix}=(I-A^4)^{-1}(A^3B+A^2B-AB-B)=(I-A^{4})^{-1}(A^2-I)(I+A)B.
\]
By plugging the matrices $A$ and $B$ given by \eqref{agent-model-ns} into above equation, we get
\begin{equation}\label{initial-even-ns}
\begin{bmatrix}
x_i(0)\\
v_i(0)
\end{bmatrix}=\begin{bmatrix}
\tfrac{1}{2a} \\
-\tfrac{1}{2a}
\end{bmatrix}, \quad i \in S_e,
\end{equation}
Similarly, in order to have
\[
\begin{bmatrix}
x_i(0)\\
v_i(0)
\end{bmatrix}=\begin{bmatrix}
x_i(4)\\
v_i(4)
\end{bmatrix}, \quad i \in S_o,
\]
we need
\begin{equation}\label{initial-odd-ns}
\begin{bmatrix}
x_i(0)\\
v_i(0)
\end{bmatrix}=\begin{bmatrix}
-\tfrac{1}{2a} \\
\tfrac{1}{2a}
\end{bmatrix}, \quad i \in S_o.
\end{equation}

{\bf Step 2:} In this step, we show that four inequalities \eqref{ns-group1} and \eqref{ns-group2}
for each agent $i \in \mathcal{V}$, can be reduced to
only two inequalities.

Consider the input for an agent $j \in S_o$,
\begin{align}
u_j(k) &=\alpha \sum _{i \in \mathcal{N}_j} a_{ij}(x_i(k)-x_j(k))+\beta \sum _{i \in \mathcal{N}_j} a_{ij}(v_i(k)-v_j(k)) \nonumber \\
&=\alpha \sum _{i \in \mathcal{N}_j \cap S_e} a_{ij}(x_i(k)-x_j(k))+\beta \sum _{i \in \mathcal{N}_j \cap S_e} a_{ij}(v_i(k)-v_j(k)) \nonumber \\
&\hspace*{0.4cm}+\alpha \sum _{i \in \mathcal{N}_j \cap S_o} a_{ij}(x_i(k)-x_j(k))+\beta \sum _{i \in \mathcal{N}_j \cap S_o} a_{ij}(v_i(k)-v_j(k)). \label{keyineq2-ns}
\end{align}
Taking into account that $x_i(0)=x_j(0)$ and $v_i(0)=v_j(0)$ for all $i,j \in S_o$ implied by \eqref{initial-odd-ns}, from \eqref{keyineq2-ns},
we obtain
\begin{equation}\label{keyineq1-ns}
u_j(k) =\alpha \sum _{i \in \mathcal{N}_j \cap S_e} a_{ij}(x_i(k)-x_j(k))+\beta \sum _{i \in \mathcal{N}_j \cap S_e} a_{ij}(v_i(k)-v_j(k)).
\end{equation}
Similarly, we can show that for agent $i \in S_e$,
\begin{equation*}
u_i(k) =\alpha \sum _{j \in \mathcal{N}_i \cap S_o} a_{ij}(x_j(k)-x_i(k))+\beta \sum _{j \in \mathcal{N}_i \cap S_o} a_{ij}(v_j(k)-v_i(k)).
\end{equation*}

From \eqref{keyineq1-ns} and the initial states given by \eqref{initial-even-ns}  and \eqref{initial-odd-ns},
it is easy to verify that for $j \in S_o$, $u_j(k+2) \geq 1$  are equivalent to  $u_j(k) \leq -1$ for $k=0,1$.
Similarly for $i \in S_e$, $u_j(k+2) \leq -1$ is equivalent to $u_i(k) \geq 1$ for $k=0,1$.
Thus, the inequalities \eqref{ns-group1} and \eqref{ns-group2} are equivalent to the following inequalities
\begin{equation}\label{even-together}
\begin{system}{ccl}
u_i(0) & \geq & 1 \\
u_i(1) & \geq & 1, \quad i \in S_e,
\end{system}
\end{equation}
and
\begin{equation}\label{odd-together}
\begin{system}{ccl}
u_j(0) & \leq & -1 \\
u_j(1) & \leq & -1, \quad j \in S_o,
\end{system}
\end{equation}

{\bf Step 3:} It is clear that
if for each edge $(i,j) \in \mathcal{E}$, where $i \in S_e$ and $j \in S_o$
the following two conditions
\begin{equation}\label{key-ineq-ns}
\begin{system}{ccl}
a_{ij} \Bigl\{ \alpha\left[x_i(0)-x_j(0)\right]+\beta\left[v_i(0)-v_j(0)\right]\Bigr\}&=a_{ij}\frac{\alpha-\beta}{a} \leq -1, \\
a_{ij}\Bigl\{ \alpha\left[x_i(1)-x_j(1)\right]+\beta\left[v_i(1)-v_j(1)\right]\Bigr\}&=a_{ij}\frac{-\alpha-\beta}{a} \leq -1,
\end{system}
\end{equation}
are satisfied then \eqref{even-together} and \eqref{odd-together} are satisfied.
It is easy to see that if the feedback gain parameters $\alpha$ and $\beta$ satisfy \eqref{condition-ns},
then \eqref{key-ineq-ns} hold for each edge $(i,j) \in \mathcal{E}$, where $i \in S_e$ and $j \in S_o$.

From the above analysis, we see that the multi-agent system \eqref{agent-model-ns} under
the distributed controller \eqref{controller-di} with the feedback gain parameters
$\alpha$ and $\beta$ satisfying the condition \eqref{condition-ns} exhibits a periodic behavior with period
$T=4$ if the initial states 
satisfy the conditions \eqref{initial-even-ns} and \eqref{initial-odd-ns}.
\end{proof}

\begin{remark}
We note that the periodic behavior presented in Theorem \ref{thm2} exists for a particular set of
initial states as given by \eqref{initial-even-ns} and \eqref{initial-odd-ns}.
The period is $T=4$, which does not depend on the network topology,
while the feedback gain parameters for achieving this periodic behavior depend on the network topology as given by
\eqref{condition-ns}. This is in contrast to the double integrator case since it is noted in Remark \ref{remark-double}
that the feedback gain parameters for achieving periodic behaviors do not depend on the network topology, however,
the periodic $T$ depends on the network topology.
\end{remark}

\section{Illustrative Examples}\label{sec-example}
In this section, we present two examples to illustrate the result.
In both examples, the network consists of $N=7$ agents and the
topology is given by the undirected weighted graph depicted
in Figure ~\ref{fig1}.
\begin{figure}[ht]
  \begin{center}
    \includegraphics[scale=0.4]{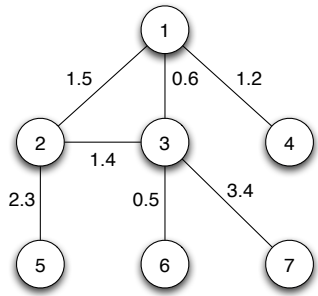}
  \end{center}
  \caption{Network with seven agents}
  \label{fig1}
\end{figure}

\subsection{Double Integrator Case}
We first consider the multi-agent model \eqref{agent-model-di} with the feedback law \eqref{controller-di}.
Choose the feedback gain parameters $\alpha=0.4$ and $\beta=0.42$ so that \eqref{period} is satisfied.
It is easy to see that $\bar{a}=a_{36}=0.5$, and therefore we choose $m=11$ such that
\eqref{suffm} is satisfied.
From the proof of Theorem \ref{thm1}, we see that
the multi-agent system exhibits a periodic solution of period $T=22$ if
the initial states satisfy \eqref{initial-cond-v}, \eqref{cyclic-cond} and \eqref{xxxx} with $m=11$, i.e.,
$v_i(0)=-5.5$ for $i \in S_e=\{1,5,6,7\}$ and $v_i(0)=5.5$ for $i \in S_o=\{2,3,4\}$ and
\begin{equation}\label{key-initial-cond}
\begin{system}{ccl}
x_2(0) = x_3(0),\\
2.1167 \leq x_1(0)-x_2(0) \leq 8.8833,  \\
4.6167 \leq x_1(0)-x_3(0) \leq 6.3833, \\
2.5333 \leq x_1(0)-x_4(0) \leq 8.4667,  \\
1.5370 \leq x_5(0)-x_2(0) \leq 9.4630,  \\
5.4500 \leq x_6(0)-x_3(0) \leq 5.5500,  \\
1.1852 \leq x_7(0)-x_3(0) \leq 9.8147.
\end{system}
\end{equation}
We then choose
\[
x_1(0)=21,\, x_2(0)=16.02, \, x_3(0)=16.02, \, x_4(0)=15, \, x_5(0)=20, \, x_6(0)=21.5, \, x_7(0)=18,
\]
so that the above condition is satisfied.

With these initial states, The multi-agent system exhibits periodic behavior of period $22$
as shown in Figure ~\ref{periodic22}.
To make the figure more clear, we have only included
state trajectories for agents $1$, $2$, $5$ and $7$ in Figure ~\ref{periodic22}.
State trajectories for agent $1$ are also given in Figure ~\ref{fig:xv-non-tree}.

Note that if the network topology is a tree, for example, if
there is no edge between agent $2$ and agent $3$ in Figure ~\ref{fig1},
we do not need to choose $x_2(0)=x_3(0)$ in order to generate the periodic solutions.
This can be seen by noticing that the last two terms in \eqref{keyineq2} are vanishing,
thus \eqref{keyineq1} is satisfied for all $x_2(0)$ and $x_3(0)$.
\begin{figure}[ht]
\begin{center}
\includegraphics[scale=0.5]{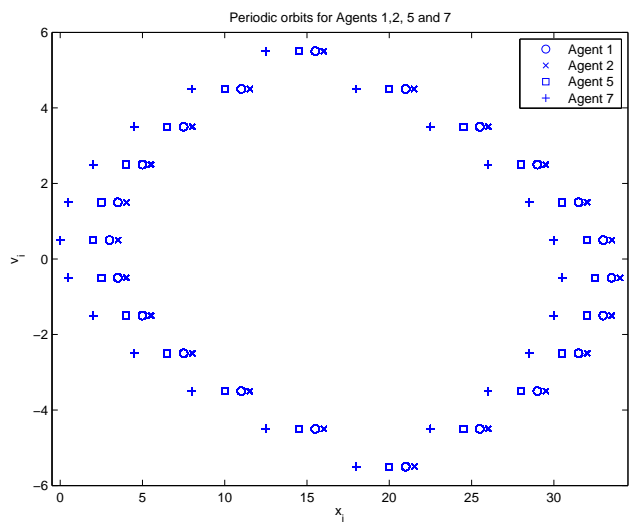}
\end{center}
\caption{Periodic behavior of period $22$}
\label{periodic22}
\end{figure}

\begin{figure}[ht]
\begin{center}
\subfigure
{
\label{fig:di-position}
\includegraphics[scale=0.35]{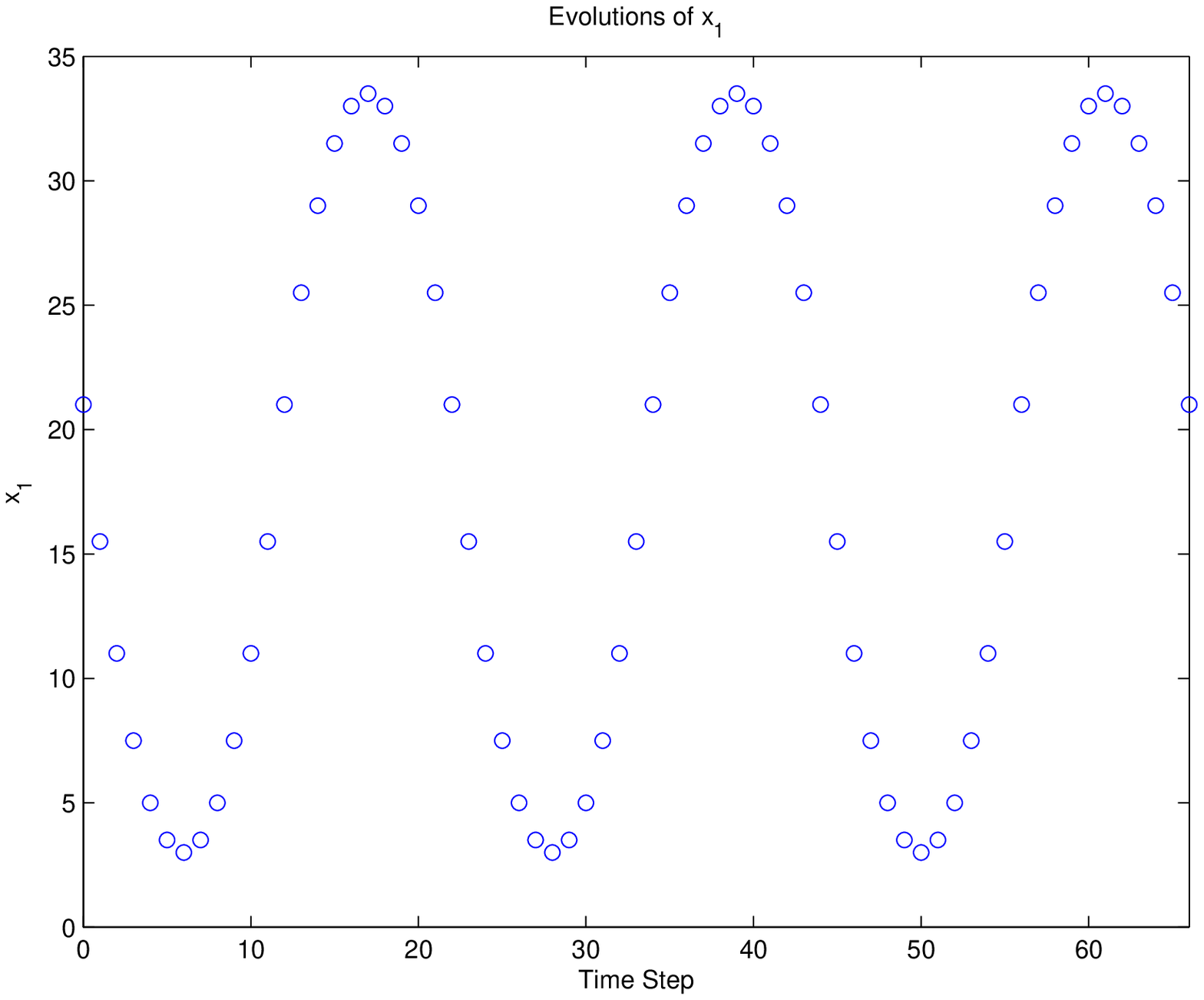}
}
{
\label{fig:di-velocity}
\includegraphics[scale=0.35]{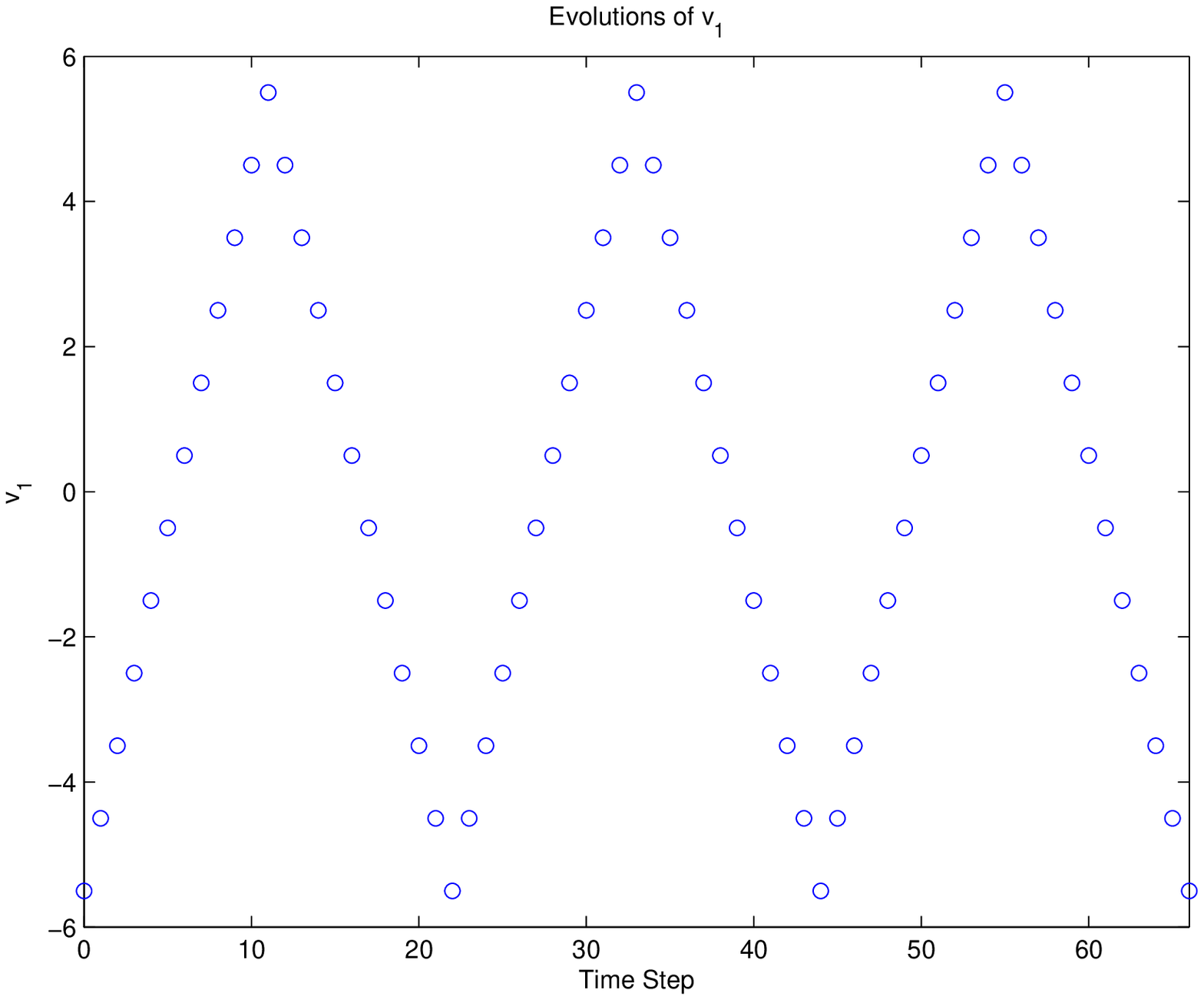}
}
\end{center}
\caption{Trajectories of $x_1$ and $v_1$}
\label{fig:xv-non-tree}
\end{figure}


\subsection{Neutrally Stable System Case}
Next, we consider the multi-agent model \eqref{agent-model-ns} with the feedback law \eqref{controller-di}.
Let $a=\frac{1}{2}$, and choose the feedback gain parameters $\alpha=-0.5$ and $\beta=2$ so that \eqref{condition-ns} is satisfied.
From the proof of Theorem \ref{thm2}, we see that
the multi-agent system exhibits a periodic solution of period $T=4$ if
the initial states satisfy \eqref{initial-even-ns} and \eqref{initial-odd-ns}, that is,
$x_i(0)=1, \, v_i(0)=-1$ for $i \in S_e=\{1,5,6,7\}$ and $x_i(0)=-1, \, v_i(0)=-1$ for $i \in S_o=\{2,3,4\}$.
Figure ~\ref{periodic-ns} shows that the multi-agent system exhibits a periodic behavior with $T=4$.
State evolutions for agent $1$ are given in Figure ~\ref{fig:xv-ns}.

\begin{figure}[ht]
\begin{center}
\includegraphics[scale=0.5]{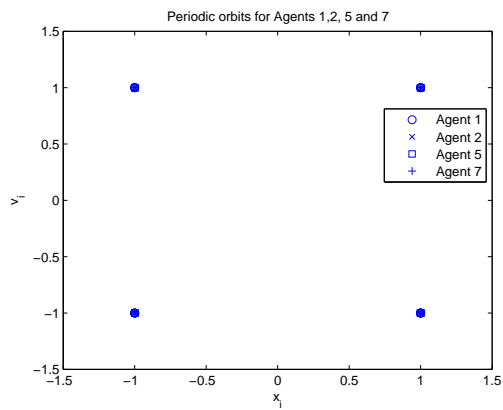}
\end{center}
\caption{Periodic behavior of period $4$}
\label{periodic-ns}
\end{figure}

\begin{figure}[ht]
\begin{center}
\subfigure
{
\label{fig:ns-position}
\includegraphics[scale=0.35]{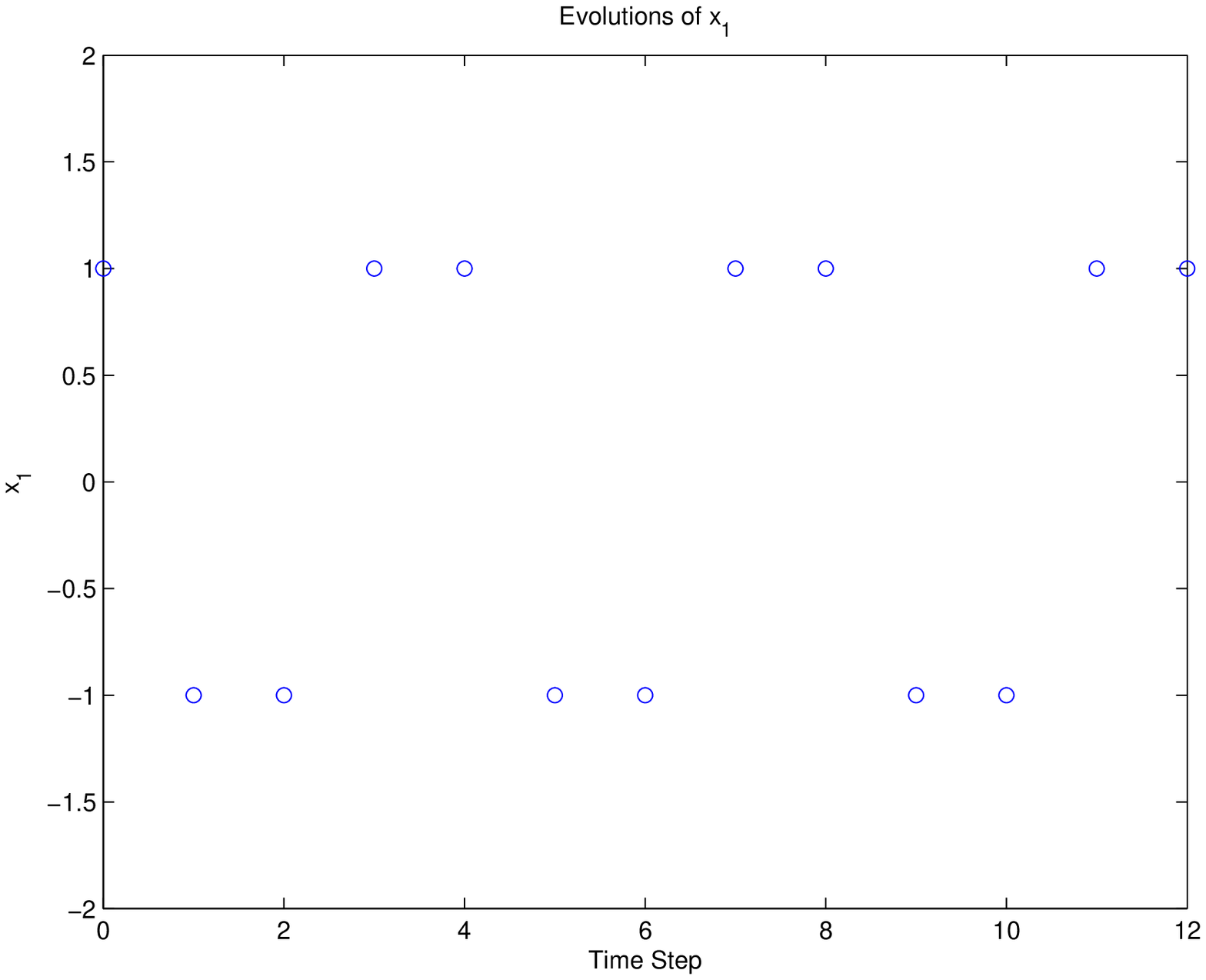}
}
{
\label{fig:ns-velocity}
\includegraphics[scale=0.35]{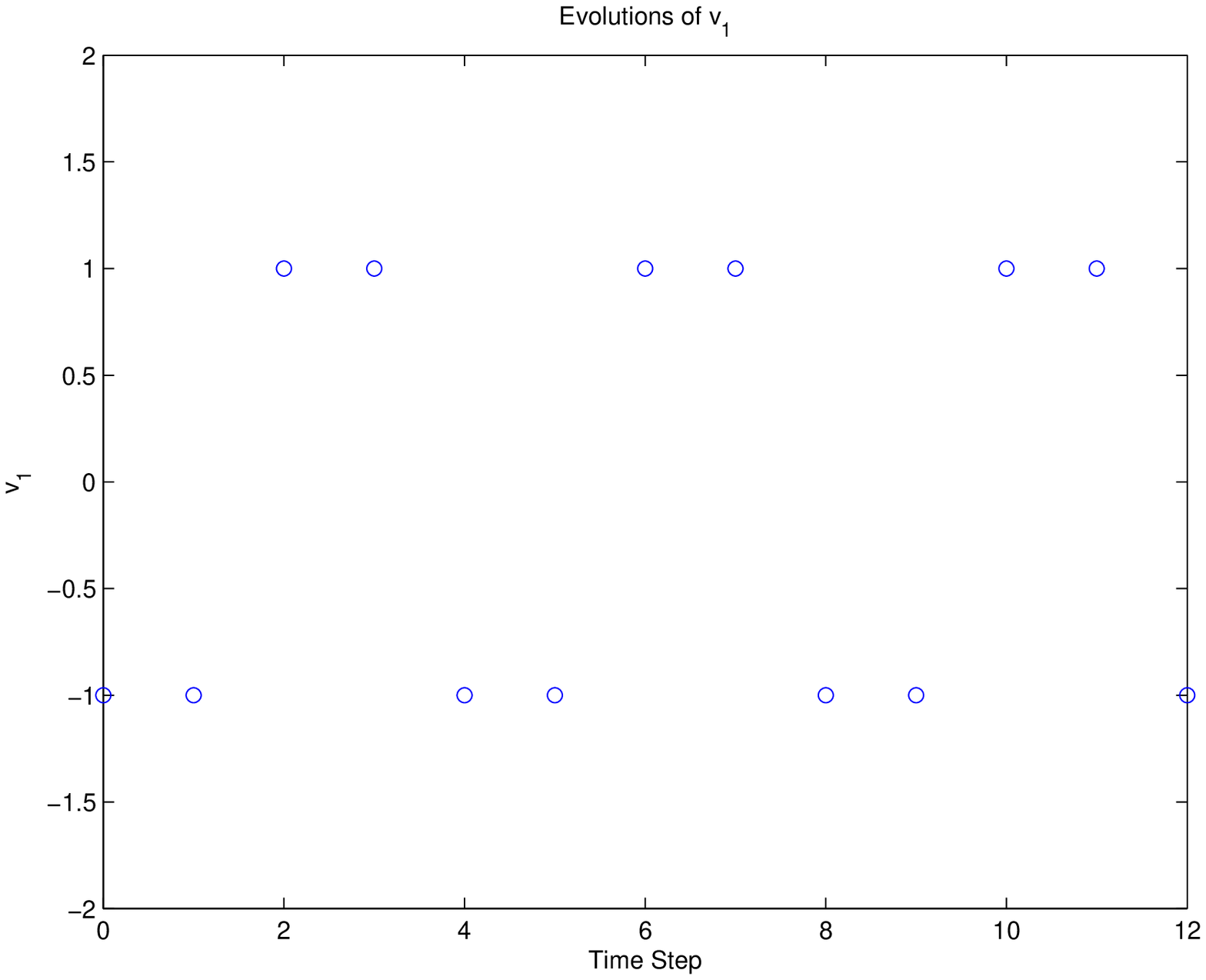}
}
\end{center}
\caption{Trajectories of $x_1$ and $v_1$}
\label{fig:xv-ns}
\end{figure}

\section{Conclusion}\label{sec-conclusion}
In this paper, we present two constrained multi-agent models, which generate
periodic behaviors for a particular set of initial states if the feedback gain parameters satisfy a certain condition.
In one model, the period depends on the network topology, however, the feedback gain parameters for achieving this are independent
of the network topology, while in the other model, the period is $T=4$, independent of the network topology,
however, the feedback gain parameters for achieving this depends on the network topology.
Whether periodic behaviors are stable and whether periodic behaviors exist for other initial states are currently under investigation.

\end{document}